\documentclass[11pt,a4paper]{amsart}
\usepackage{amsaddr}
\usepackage{graphicx}
\usepackage{tikz}
\usetikzlibrary{decorations.pathmorphing}
\usepackage{amsmath}
\usepackage{amssymb}
\usepackage{amsthm}
\usepackage{color}

\usepackage{hyperref}
\usepackage{bm}
\usepackage{dsfont}
\usepackage{multirow}

\newcommand{\tr}[1]{\text{tr}\{{#1}\}}

\newcommand{\K}{\mathcal{K}}
\renewcommand{\H}{\mathcal{H}}

\newcommand{\N}{\mathbb{N}}

\newcommand{\R}{\mathbb{R}}
\newcommand{\C}{\mathbb{C}}

\newcommand{\Z}{\mathbb{Z}}

\newcommand{\E}{\mathbb{E}}
\newcommand{\unit}{\mathds{1}}

\newcommand{\bp}{\begin{proof}}
\newcommand{\ep}{\end{proof}}
\newcommand{\bdp}{\begin{dproof}}
\newcommand{\edp}{\end{dproof}}
\newcommand{\ra}{\rightarrow}
\renewcommand{\O}{\operatorname{O}}

\newcommand{\Sp}{\operatorname{Sp}}
\renewcommand{\sp}{\mathfrak{sp}}

\newcommand{\U}{\operatorname{U}}

\newcommand{\Mat}{\operatorname{M}}

\newcommand{\rmd}{\operatorname{d}}
\newcommand{\rmi}{\operatorname{i}}

\newcommand{\diag}{\operatorname{diag}}

\newtheorem{theorem}{Theorem}

\theoremstyle{definition}
\theoremstyle{definition}
\theoremstyle{definition}
\theoremstyle{definition}
\theoremstyle{definition}
\theoremstyle{definition}
\theoremstyle{definition}

\begin{document}
\title{Dynamical Decoupling and Homogenization of continuous variable systems}
\author{Christian Arenz, Daniel Burgarth} 
\address{Institute of Mathematics, Physics, and Computer Science\\ Aberystwyth University\\ Aberystwyth SY23 2BZ, UK}
\author{Robin Hillier}
\address{Department of Mathematics and Statistics\\ Lancaster University\\ Lancaster LA1 4YF, UK}
%\author{Daniel Burgarth}
%\address{Institute of Mathematics, Physics, and Computer Science\\ Aberystwyth University\\ Aberystwyth SY23 2BZ, UK} 
\date{\today}
\thanks{D.B. acknowledges support from the EPSRC Grant No. EP/M01634X/1 and fruitful discussions with Markus Kirschmer and Gunther Dirr.}

%Known as decoherence, the unavoidable interaction of a quantum system with its surrounding environment is usually considered to be detrimental for quantum information processing.
\begin{abstract}
For finite-dimensional quantum systems, such as qubits, a well established strategy to protect such systems from decoherence is dynamical decoupling.  However many promising quantum devices, such as oscillators, are infinite dimensional, for which the question if dynamical decoupling could be applied remained open. Here we first show that not every infinite-dimensional system can be protected from decoherence through dynamical decoupling. Then we develop dynamical decoupling for continuous variable systems which are described by quadratic Hamiltonians. We identify a condition and a set of operations that allow us to map a set of interacting harmonic oscillators onto a set of non-interacting oscillators rotating with an averaged frequency,  a procedure we call homogenization. Furthermore we show that every quadratic system-environment interaction can be suppressed with two simple operations acting only on the system. Using a random dynamical decoupling or homogenization scheme, we develop bounds that characterize how fast we have to work in order to achieve the desired uncoupled dynamics. This allows us to identify how well homogenization can be achieved and decoherence can be suppressed in continuous variable systems. 
\end{abstract}
\maketitle

\section{Introduction} 
Dynamical decoupling is a highly successful strategy to protect quantum systems from decoherence \cite{ErrorC}. Its particular strength is that it is applicable even if the details of the system-environment coupling are unknown.  Historically dynamical decoupling dates back to pioneering work in nuclear magnetic resonance (NMR) by U. Haeberlen and J. S. Waugh \cite{AvHamTheory}. In order to increase the resolution in NMR spectroscopy, pulse sequences were developed that coherently average out unwanted interactions \cite{AvHamTheory2}. Prominent examples are spin-echo techniques, such as the famous Hahn echo \cite{HahnEcho}, allowing us to measure relaxation times through applying a sequence of rotations on a spin and detecting the echo signal. In the context of suppressing decoherence and quantum information theory, the theoretical framework was developed by L. Viola and S. Lloyd \cite{LVioal1, LVioalRandD} in the late 90's. Over the years the efficiency of various decoupling schemes was studied and improved for several environmental models in \cite{DecouplingSpinB1, DecouplingSpinB2, DecouplingSpinBRand, Uhrig1, Uhrig2, Uhrig3}. Many experiments, such as \cite{NVcentre, classicalnoise, QubitinSolid}, demonstrate the applicability of dynamical decoupling in an impressive way by prolonging coherence times a few orders of magnitude. Additionally dynamical decoupling can be combined with the implementation of quantum gates, which makes it a viable option to error correction \cite{DCandEC, DCQgates}. However, when it comes to infinite-dimensional quantum systems, such as quantum harmonic oscillators, a general framework for dynamical decoupling is missing in the literature. A first step towards this direction was done in \cite{Vitali} by investigating a specific system-environment model and identifying an operation that allows to suppress decoherence.

In this article we consider a broader class of continuous variable systems by investigating dynamical decoupling for systems that are described by Hamiltonians quadratic in the operators $\hat{x}$ and $\hat{p}$ \cite{GaussianStates}. These Hamiltonians are of particular importance since they describe a wide range of continuous variable systems and their main sources of decoherence. For instance, quadratic Hamiltonians describe linear quantum optical systems \cite{QuadHamOpticalS} with applications in optical quantum computing \cite{OpticalQuantumC} and quantum metrology \cite{Metrology}, the vibrational modes of an ion chain with nearest-neighbour interactions \cite{IionChainRevervoir} and in general harmonic crystals \cite{Rubin}. Moreover, many continuous variable systems can be described by quadratic Hamiltonians in a certain approximate regime, for example opto-mechanical systems and nano-mechanical oscillators by linearizing interactions \cite{QuadHamOpto}. 

Our article is organized as follows. We start with the question of the existence of decoupling in finite and infinite dimensions. We then introduce dynamical decoupling for quadratic Hamiltonians and arrive at procedures we call \emph{homogenization} and \emph{decoherence suppression for environment-coupling}. The desired dynamics is achieved by rapidly swapping coordinates or rapidly rotating the system, respectively, see the blue arrows in Figure \ref{fig:oscillators}.
 Again we discuss their existence for given quantum systems. Finally we introduce a randomized scheme for these procedures and derive explicit analytic approximations and bounds for the gate error, i.e., the discrepancy from the idealized time-evolution of infinitely fast operations. The proofs are very technical and therefore deferred to the appendix. We illustrate and confirm the usefulness of these formulae with typical numerical examples.
\begin{figure}[!h]
	\centering
	 \includegraphics[width=0.79\columnwidth]{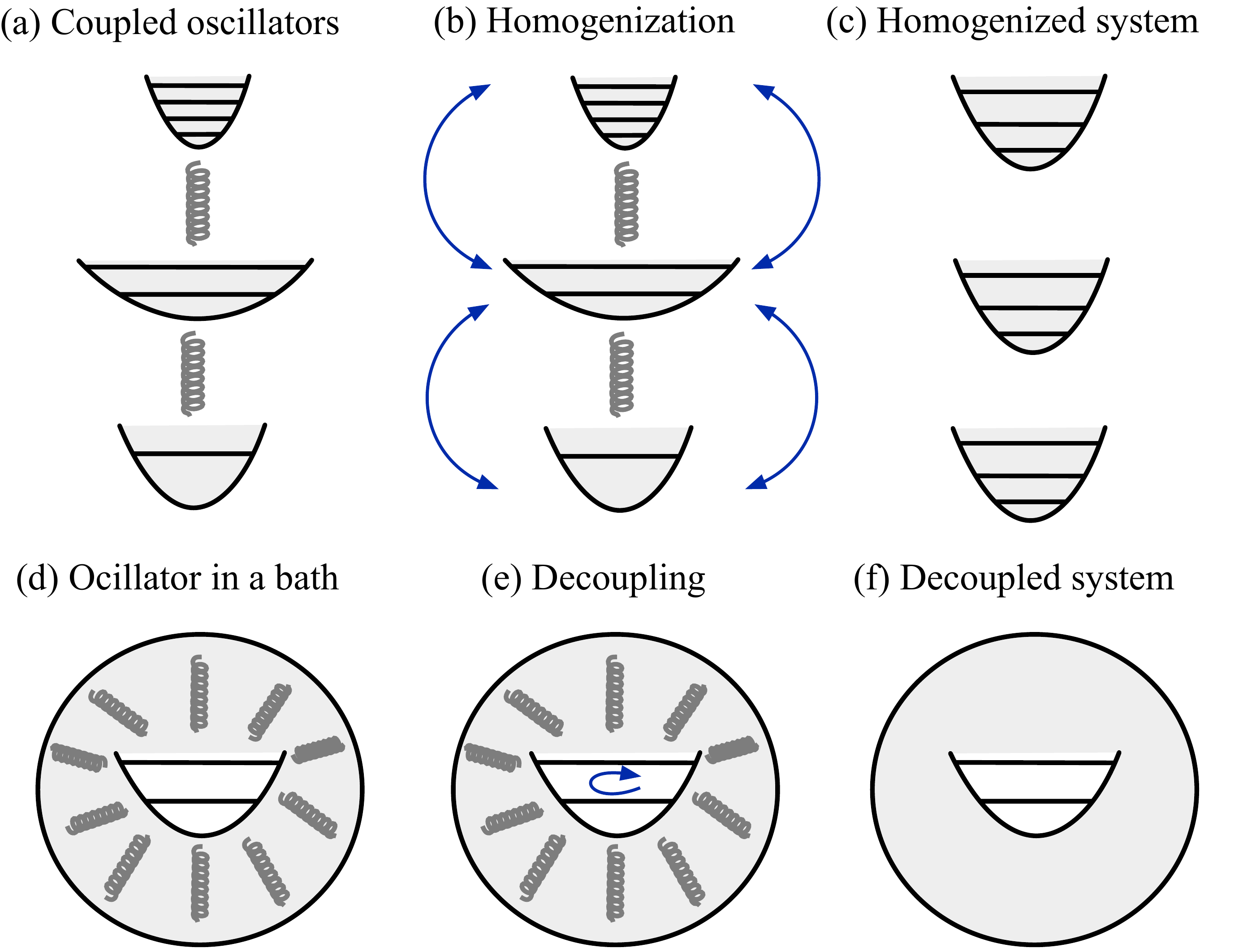}
	 \caption{\label{fig:oscillators}Illustration of the homogenization procedure of a system of coupled oscillators (a $\rightarrow$ c, see Theorem \ref{th3}) and suppression of decoherence of a single oscillator in a heat bath (d $\rightarrow$ f, see Theorem \ref{th4}).
}
\end{figure}

\section{Dynamical decoupling for finite and for infinite-dimensional systems}

Before we develop dynamical decoupling for quadratic Hamiltonians in infinite-dimensional quantum systems, we first review the concept of dynamical decoupling for finite-dimensional quantum systems, focussing on the group-based approach \cite{ErrorC, USMath}. Consider a finite-dimensional quantum system, say of dimension $n\in\N$, with Hilbert space $\C^n$ and Hamiltonian $H$. The idea of dynamical decoupling is to rapidly rotate the quantum system by means of classical fields in order to average the system-environment coupling in $H$ to zero. This can be achieved by modifying the dynamics applying unitary \emph{decoupling operations} $v$ instantaneously in time steps $\tau$, which are taken from a decoupling set. Such a \emph{decoupling set} $V$ is a finite group of unitary $n\times n$ matrices such that, for every $x\in\Mat_n(\C)$,
\begin{align}\label{eq:decouplingcondfinitedim}
\Pi_0(x):= \frac{1}{|V|}\sum_{v\in V}v^{\dagger} x v = \lambda \unit_n,	
\end{align}
with some $\lambda\in\C$ depending on $x$. It can easily be seen that $\lambda=\frac{1}{n}\tr{x}$. In particular, for traceless $x$ we have $\Pi_0(x)=0$.
These decoupling operations can be applied according to a fixed deterministic scheme or randomly from the set $V$. Let us look for simplicity at the fixed scheme. We then get the modified dynamics of time-evolution up to time $t=|V|\tau$:
\[
\prod_{v\in V} v^\dagger e^{-\rmi \tau H} v = \unit_n - \rmi \tau
\sum_{v\in V}v^{\dagger}Hv + O((\tau \|H\|_\infty)^2),
\]
where $\|\cdot\|_\infty$ denotes the standard operator norm on $\Mat_n(\C)$.
Condition \eqref{eq:decouplingcondfinitedim} guarantees that $\Pi_0(H)=\frac{1}{n}\tr{H}\unit_n$ and ensures the cancellation of the modified unitary time-evolution operator in first order in $\tau\| H\|_\infty$, while higher orders can be neglected under this assumption. Therefore, for $\tau\|H\|_\infty \ll 1$, the modified time evolution becomes up to a global phase effectively the identity. We refer to \cite{ErrorC, USMath} for a detailed definition of dynamical decoupling. 

The first question here is when a decoupling set can actually be found. 

\begin{theorem}\label{th1}
For every quantum system of finite dimension $n$, there exists a decoupling set, for example
\[
V= \U(n,\{0,1,-1,+\rmi,-\rmi\}),
\]
the group of $n\times n$ unitary matrices with entries in $\{0,1,-1,+\rmi,-\rmi\}$.
\end{theorem}

Notice that a decoupling set for a given quantum system is not unique, in general.

\begin{proof}
We first notice from \eqref{eq:decouplingcondfinitedim} that, for every $x\in\Mat_n(\C)$, $\Pi_0(x)$ lies in the commutant of $V$, which follows from the group property and finiteness of $V$. Thus if $V\subset \U(n)$ acts irreducibly on $\C^n$ then this commutant consists of $\C\unit_n$ alone. We therefore have to construct $V$ that acts irreducibly. We claim that $V= \U(n,\{0,1,-1,+\rmi,-\rmi\})$
does this. The elements in $V$ are those unitary matrices which have exactly one non-zero entry per row and per column, and this entry is either $+1,-1,+\rmi$ or $-\rmi$. 
It is easily seen that $V$ forms a finite group. Given an arbitrary $\xi\in\C^n$, with $i$-th component $\xi_i$ non-zero, we can take $v$ the diagonal matrix with $+1$ on all diagonal entries except for $-1$ at the $i$-th position. Then $\unit_n,v\in V$ and $\unit_n\xi-v\xi = 2\xi_i e_i$, where $e_i\in\C^n$ is the $i$-th canonical basis vector. All permutation matrices of $\C^n$ (matrices that permute coordinate entries w.r.t. the canonical basis of $\C^n$) are contained in $V$, and so applying a suitable permutation matrix $w\in V$ to $2\xi_i e_i$, we can get a multiple of an arbitrary basis vector $e_j$, $j=1,\ldots,n$, hence every vector in $\C^n$ through linear combination. This shows that $V$ acts irreducibly and transitively on $\C^n$, and thus proves our claim.
\end{proof}

Let us look at some simple examples. For a single qubit such a decoupling set may be chosen more explicitly as the Pauli group $\{\alpha\mathds{1},\alpha\sigma_{x},\alpha\sigma_{y},\alpha\sigma_{z}: \alpha=\pm 1,\pm\rmi\}$, whereas for a $q$-qubit system, it may be chosen as $4^{q+1}$ combinations of the Pauli spin operators on the tensor factors. The size of the decoupling set scales exponentially with the number of qubits of the considered system. This makes deterministic decoupling schemes inefficient for large quantum systems, since the modified dynamics is obtained through taking the decoupling operations one after the other deterministically from $V$.

Clearly, in order to suppress decoherence that is induced by a generic system-environment Hamiltonian 
\begin{equation}\label{eq:genericH}
H=H_{S}+H_{E}+\sum_{\alpha}S_{\alpha}\otimes E_{\alpha},
\end{equation}
with $S_{\alpha}$ and $E_{\alpha}$ being hermitian system and environmental operators respectively, it is enough to act with $v$ on the system alone. For traceless system operators, in the limit of infinitely fast decoupling we obtained a dynamics that is decoupled from the environment, i.e. $\Pi_0(H)=\mathds{1}\otimes H_{E}$. This makes it possible to fully suppress decoherence, independently of the specific form of the system-environment interaction that is present. As long as the system is finite-dimensional, we note that this is even true in the case of infinite-dimensional environments described by some unbounded operators and $H$ as in \eqref{eq:genericH}, if we make certain plausible domain assumptions \cite{USPRA}.

The question arises if, similarly to the finite-dimensional case, in an arbitrary infinite-dimensional setting decoherence can always be suppressed through dynamical decoupling. Before we address this question, we begin with an example \cite{Vitali}, which is, to the best of our knowledge, the only study in the literature where dynamical decoupling is investigated for a specific infinite-dimensional system. 
The model that was considered in \cite{Vitali} consists of a single harmonic oscillator that interacts with an environment of harmonic oscillators. 
The system-environment interaction is given by $H_{S,E}=\sum_{k}g_{k}(ab_{k}^{\dagger}+a^{\dagger}b_{k})$ where $a,a^{\dagger}$ and $b_{k},b_{k}^{\dagger}$ are bosonic creation and annihilation operators of the system and the environmental oscillators respectively and $g_{k}$ are real coupling constants determining the strength of the interaction with each environmental oscillator. We note that $H_{S,E}$ can be brought into the form \eqref{eq:genericH} by expressing the annihilation and creation operators through $\hat{x}$ and $\hat{p}$. It was pointed out in \cite{Vitali} that through a decoupling operation $v=\exp(-i\pi a^{\dagger}a)$, which corresponds to a phase space rotation of the system oscillator around $\pi$, the sign in front of $H_{S,E}$ can be reversed. This makes it possible to suppress such a system-environment interaction if we apply the decoupling operation reasonably fast. The efficiency of this procedure was studied in detail in \cite{Vitali} in terms of spectral properties of the environment. 
For our purposes what is important is the observation that such an interaction can always be suppressed with a single operation. 

We then may ask if we can identify a decoupling set that allows us to suppress decoherence for arbitrary infinite-dimensional systems. The most natural infinite-dimensional meaning of decoupling would be the following: for an infinite-dimensional quantum system with Hilbert space $\H$ and bounded operators $B(\H)$, a decoupling set is a finite subgroup $V\subset \U(\H)$ such that 
\begin{equation}\label{eq:dec-inf}
\Pi_0(x) = \frac{1}{|V|}\sum_{v\in V}v^{\dagger} x v = \lambda \unit_\H,
\end{equation}
with some $\lambda\in\C$ depending on $x$, is satisfied for all $x\in B(\H)$.

\begin{theorem}\label{th0}
An infinite-dimensional quantum system has no decoupling set.
\end{theorem}

\begin{proof}
Consider an infinite-dimensional system with Hilbert space $\H$ and suppose there is a decoupling set $V$. Then considering a rank-one projection $x\in B(\H)$, we see that $v^\dagger x v$ has rank 1  again, and hence that the left-hand side of \eqref{eq:dec-inf} has rank at least 1 (because sum of positive elements) and at most $|V|<\infty$ (because each of the summands has rank 1). On the other hand, the right-hand side has either rank 0 (if $\lambda=0$) or $\infty$. This is a contradiction, so no such $V$ exists.
\end{proof}

Let us make a few remarks:
\begin{itemize}
\item[(i)] Without contradicting Theorem \ref{th0}, it could still be possible to achieve $\Pi_0(H)=\lambda\unit_\H$ for specific Hamiltonians $H$ rather than all $x\in B(\H)$.
\item[(ii)] It can be shown \cite{PhDThesis} that interactions containing system operators with an unbounded positive spectrum can never be suppressed through applying unitary decoupling operations infinitely fast.
\item[(iii)] It seems worth studying an infinite compact version of decoupling sets, which is work in progress.
\end{itemize}
As dynamical decoupling does not work for generic infinite-dimensional quantum systems, we decided to investigate in more detail the specific class of quadratic Hamiltonians and to adjust the decoupling condition to suit that setting in a meaningful way. This allows us to represent the dynamics by a symplectic transformation on a \emph{finite-dimensional} space. Additionally this has the advantage that we can avoid the mathematical subtleties arising in infinite-dimensional spaces and the related problems in the characterization of the relevant time-scales.

%\vspace*{1cm}
\section{Dynamical decoupling for quadratic Hamiltonians}
We consider an $n$-mode bosonic system described by $n$ pairs of quadrature operators $\hat{x}_{j}$ and $\hat{p}_{j}$ acting on an infinite-dimensional Hilbert space $\H$ and satisfying the canonical commutation relation $[\hat{x}_{i},\hat{p}_{j}]=i\delta_{i,j}\unit_\H$. We always write $L(\H)$ for the linear (possibly unbounded) operators on $\H$, $B(\H)$ for the algebra of bounded operators on $\H$, $\U (\H)\subset B(\H)$ for the group of unitary operators on $\H$. 

By introducing the diagonal matrix $\boldsymbol{R}=\diag (\hat{x}_{1},...,\hat{x}_{n},\hat{p}_{1},...,\hat{p}_{n})\in \Mat_2(\R)\otimes\Mat_n(\R)\otimes L(\H)$, the commutation relation can be written as $[\boldsymbol{R}_{ii},\boldsymbol{R}_{jj}]=iJ_{i,j}\unit_\H$ where
\begin{align}
\label{eq:symplecticformJ}
J=\left(\begin{matrix}
	0 & \mathds{1}_{n}\\
	-\mathds{1}_{n} & 0 
\end{matrix}
\right),
\end{align}
is the symplectic form. 
We note that a suitable basis change leads to
\begin{align}\label{eq:basischange}
\boldsymbol{R}=\diag(\hat{x}_{1},\hat{p}_{1},...,\hat{x}_{n},\hat{p}_{n}), \quad 	J=\bigoplus_{j=1}^{n}\left(
	\begin{matrix}
	0 & 1\\
	-1 & 0	
	\end{matrix}
	\right),
\end{align}
but we will use the former choice of basis except otherwise mentioned. 

We are interested in quantum systems that are described by a quadratic Hamiltonian of the form 
\begin{align}
H=\frac{1}{2}\sum_{i,j}A_{i,j}\boldsymbol{R}_{ii}\boldsymbol{R}_{jj}\in L(\H),	
\end{align}
with $A$ being a real and symmetric $2n\times 2n$ matrix. The corresponding unitary time-evolution operations $U(t)=e^{-iHt}$ are the so-called Gaussian operations since they preserve the Gaussian character of quantum states. If we consider the Heisenberg evolution of the quadrature operator we obtain 
\[
(\unit_{2n}\otimes U(t)^{\dagger})\boldsymbol{R}(\unit_{2n}\otimes U(t))=(S(t)\otimes \unit_\H)\boldsymbol{R},
\]
where 
\begin{align}
\label{eq:symplecticdynamics}
S(t):=e^{-tAJ}, \quad t\in\R_+,	
\end{align}
belongs to the symplectic group $\Sp(2n,\R)$. Restricting to quadratic Hamiltonians, we see that there is a one-to-one correspondence between the time evolution operators $(U(t))_{t\in\R_+}$ acting on an infinite-dimensional Hilbert space and the finite-dimensional matrices $(S(t))_{t\in\R_+}$, which allows us to reduce to a finite-dimensional setup. For further details regarding quadratic Hamiltonians and symplectic transformations we refer to \cite{GaussianStates}. 

Now we want to introduce dynamical decoupling within the framework of symplectic transformations. While usually dynamical decoupling is introduced within quantum control theory using averaged Hamiltonians \cite{ErrorC} we consider here a decoupling sequence that captures the main aspects of dynamical decoupling. We remark that a formulation within averaged Hamiltonian theory in terms of symplectic transformation is straightforward \cite{PhDThesis}. Here we instantaneously apply symplectic decoupling operations from a finite group $G\subset \Sp(2n,\R)$, in time steps $\frac{t}{|G|}$, with $|G|$ being the number of elements in $G$. If this is done deterministically in a fixed periodic cycle running through all of $G$, call it $g_1,\ldots g_{|G|}$ with $g_1=g_{|G|+1}=1$, then the correction operation applied instantaneously at time $kt/|G|$ is $g_{k+1}^{-1}g_k$, and the modified dynamics at $t$ becomes  
\begin{align}\label{eq:dec-def}
	S^{(1)}(t)=\prod_{k=1}^{|G|}g_k e^{-\frac{t}{|G|} A J}g_k^{-1}=\prod_{k=1}^{|G|} e^{-\frac{t}{|G|} g_kAJ g_k^{-1}}.
\end{align}
If we repeat this cycle $N$ times and shrink the time steps by a factor $1/N$, we obtain
\[
S^{(N)}(t)=\Big(\prod_{k=1}^{|G|} e^{-\frac{t}{|G|N} g_kAJ g_k^{-1}} \Big)^N.
\]
Using a generalized Trotter formula \cite{TrotterSuzuki}, we get
 \begin{align*}
 \lim_{N\to\infty} S^{(N)}(t)=\exp\left(-t \sum_{k=1}^{|G|} g_kAJ g_k^{-1}\right) =\exp\left(-t \sum_{g\in G}gAJ g^{-1}\right).
 \end{align*}
Analogously to the unitary case we can define a tentative symplectic decoupling condition for $G$: 
\begin{align}
\label{eq:defmap}
\frac{1}{|G|}\sum_{g\in G}gAJ g^{-1}= \lambda \unit_{2n}, \quad A\in\Mat_{2n}(\R),
\end{align}
with some $\lambda\in\R$ depending on $A$. However, using the relation $J=g^{T}Jg$, valid for every $g\in\Sp(2n,\R)$, we can rewrite \eqref{eq:defmap} as
\[
\frac{1}{|G|}\sum_{g\in G}gAg^{T}= \lambda J, \quad A\in\Mat_{2n}(\R).
\]
Considering e.g. a positive (hence nonzero) $A\in\Mat_{2n}(\R)$ in this equation, we see that the left-hand side being a sum of positive operators must be positive again, whereas the right-hand side is antisymmetric, which is impossible. Thus the tentative decoupling condition \eqref{eq:defmap} cannot be realized. It reflects for exmaple the fact that a system of non-interacting harmonic oscillators cannot be stopped rotating. As already pointed out in the previous section, dynamical decoupling in the usual sense cannot work for every infinite-dimensional system, and apparently also for quadratic Hamiltonians this is not possible in this strict sense.  However, maybe we are demanding too much.
\subsection*{Homogenization}
Since we cannot decouple harmonic oscillators, a natural relaxation of condition \eqref{eq:defmap} may be written as follows:

A finite subgroup $G\subset \Sp(2n,\R)$ is called a \emph{homogenization set} if, for every symmetric $A\in\Mat_{2n}(\R)$, we have
\[
\frac{1}{|G|}\sum_{g\in G}gA J g^{-1}=\lambda J,	
\]
or equivalently
\begin{align}
\label{eq:homogenizationcond}
\Pi(A):=\frac{1}{|G|}\sum_{g\in G}gAg^{T}=\lambda \unit_{2n},	
\end{align}
with some $\lambda\in \R$ depending on $A$.

In words: instead of requiring that the system does not evolve anymore, we now require that a set of harmonic oscillators do not interact with each other but rotate all with the same frequency $\lambda$ after we have applied symplectic operations infinitely fast. We call this process \emph{homogenization}. We note here that the homogenization procedure is similar to symmetrization of quantum states developed in \cite{SymmetrizationStates}. It remains to identify a set $G$ of symplectic operations that satisfies the homogenization condition \eqref{eq:homogenizationcond}. 

\begin{theorem}\label{th3}
For every $n\in\N$, a homogenization set for $\Mat_{2n}(\R)$ exists, for example
\begin{align}
\label{eq:homogenisationset}
G=\left< \unit_2\otimes \O(n,\mathbb Z),J \right >.
\end{align}
\end{theorem}

\begin{proof}
First of all, it is clear that $\unit_2\otimes \O(n,\mathbb Z)$ and $J$ commute, so $gJg^T=J$, for every $g\in \unit_2\otimes \O(n,\mathbb Z)$, and $JJJ^T=J$. Thus $G\subset \Sp(2n,\R)$.
We decompose $A$ according to the product $\Mat_2(\R)\otimes\Mat_n(\R)$ into four $n\times n$ blocks $A^{(i,j)},~i,j=1,2$, where $(A^{(1,2)})^{T}=A^{(2,1)}$ and $(A^{(i,i)})^{T}=A^{(i,i)}$. We can then write \eqref{eq:homogenizationcond} as
\begin{align}
\label{eq:map1}
\Pi(A)=\frac{1}{|G|}&\sum_{g\in O(n,\mathbb Z)}\left[ \left(\begin{matrix}
gA^{(1,1)}g^{T} & gA^{(1,2)}g^{T}\\
gA ^{(1,2)}g^{T} & gA^{(2,2)}g^{T}	
\end{matrix} \right)\right.\nonumber \\  
    &+J \left. \left(\begin{matrix}
gA^{(1,1)}g^{T} & gA^{(1,2)}g^{T}\\
gA ^{(2,1)}g^{T} & gA^{(2,2)}g^{T}	
\end{matrix}
    \right)   J^{T}
       \right].	
\end{align}
Now we recall from the proof of Theorem \ref{th1} that the group $\O(n,\Z) = \U(n,\{0,+1,-1\})$ acts irreducibly on $\R^n$, so it forms a decoupling set for $\Mat_n(\R)$ in the sense of \eqref{eq:decouplingcondfinitedim}. Applying \eqref{eq:decouplingcondfinitedim} to each block, we find 
\begin{align*}
\Pi(A)&=\frac{1}{2}\left(\begin{matrix}
\frac{\tr{A^{(1,1)}}}{n}\mathds{1}_{n} & \frac{\tr{A^{(1,2)}}}{n}\mathds{1}_{n}\\
\frac{\tr{A^{(2,1)}}}{n}\mathds{1}_{n} & \frac{\tr{A^{(2,2)}}}{n}\mathds{1}_{n}
\end{matrix}
    \right)\nonumber\\
    &+\frac{1}{2}J  \left(\begin{matrix}
\frac{\tr{A^{(1,1)}}}{n}\mathds{1}_{n} & \frac{\tr{A^{(1,2)}}}{n}\mathds{1}_{n}\\
\frac{\tr{A^{(2,1)}}}{n}\mathds{1}_{n} & \frac{\tr{A^{(2,2)}}}{n}\mathds{1}_{n}	
\end{matrix}
    \right)   J^{T},\nonumber \\
    &=\frac{\tr {A^{(1,1)}}+\tr {A^{(2,2)}}}{2n}\mathds{1}_{2n}\\
    &=\frac{\tr A}{2n}\mathds{1}_{2n},\nonumber	
\end{align*}
which completes the proof and we identify $\lambda=\frac{\tr{A}}{2n}$ as an averaged frequency. 
\end{proof}

The elements in $\unit_2\otimes \O(n,\Z)$ are ``signed permutation matrices". In words:
through rapidly swapping the coordinates of the oscillators, we can map a set of $n$ interacting harmonic oscillators onto non-interacting oscillators rotating with an averaged frequency $\lambda=\frac{\tr{A}}{2n}$. 

\subsection*{Remark on finite-energy correction operations}

So far, the correction operations are assumed to be applied in the form of infinitely strong instantaneous pulses. This is certainly idealised and not very physical, yet generally taken for granted when discussing the theory of dynamical decoupling as it is a reasonable approximation facilitating most computations. If we explicitly want to work with finite-energy operations, a typical alternative would be so-called \emph{Eulerian dynamical decoupling} \cite{EulerianDecoupling}. In that case, the decoupling operations are implemented by a continuous path in the automorphism group of the system rather than a discontinuous path which jumps between the elements $G$. Let us describe the idea briefly, following \cite{EulerianDecoupling} but adapted to the context of homogenization. We  refer to future work for details of the construction, proofs, and further consequences.

Let $G$ be a homogenization set, i.e., a finite subgroup of $\Sp(2n,\R)$ as above, and $\Gamma$ a minimal generating set for $G$. Consider the Cayley graph $C(G,\Gamma)$, which has $G$ as vertices and $\{(g,\gamma g): g\in G, \gamma\in\Gamma\}$ as directed edges. This means that every vertex has $|\Gamma|$ incoming and $|\Gamma|$ outgoing edges. An Eulerian cycle is a cycle through this graph which travels precisely once through each edge according to its orientation, and hence $|\Gamma|$ times through each vertex, and with total length $|G|\cdot |\Gamma|$. One can prove that such a cycle exists. Let us fix a total time $t$ and write $\tau:= \frac{t}{N |G|\cdot |\Gamma|}$ for the time difference between consecutive operations, i.e., between consecutive vertices. Let now $\K$ be such an Eulerian cycle, which we write as an ordered tuple $\K=(g_1,g_2,\ldots,g_{|G|\cdot |\Gamma|})\in G^{|G|\cdot |\Gamma|}$ and we extend this periodically to $G^\N$, so $g_{|G|\cdot |\Gamma|+1}=g_1$. Then at time $k\tau$, the operation $\gamma_k := g_{k+1}^{-1}g_{k}\in\Gamma$ is applied instantaneously, so that the total time evolution in \eqref{eq:dec-def} becomes
\begin{align}\label{eq:dec-def2}
S^{(N)}(t)=\prod_{k=1}^{N |G|\cdot|\Gamma|} g_k e^{-\tau A J}g_k^{-1} 
= g_1 \Big( \prod_{k=1}^{N |G|\cdot|\Gamma|} \gamma_k e^{-\tau A J} \Big) g_1^{-1}.
\end{align}

So far, we have just modified the order of the correction operations, not their energy. Now instead suppose that the operation $\gamma_k$ is not applied instantaneously at time $k\tau$ but rather obtained as a continuous curve from $1$ to $\gamma_k$. Such a map is constructed as follows: let us recall, e.g.~from \cite{HN}, that every element $\gamma\in\Sp(2n,\R)$ can be decomposed as
\[
\gamma = e^{\frac{\tau}{2}X} e^{\frac{\tau}{2}Y},
\]
with certain $X,Y\in\sp(2n,\R)$. Then we can define
\[
\tilde{\gamma}: s\in [0,\tau] \ra \Sp(2n,\R),\quad \tilde{\gamma}(s) = \left\{ \begin{array}{l@{\: :\:}l}
 e^{ sY} & s\in[0,\tau/2]\\
e^{(s-\frac{\tau}{2})X} e^{\frac{\tau}{2}Y} & s\in(\tau/2,\tau].
\end{array}
\right.
\]
and we have $\tilde{\gamma}(0)=1$ and $\tilde{\gamma}(\tau)=\gamma$.
Then \eqref{eq:dec-def2} changes to a time-ordered product, and in first order expansion this becomes
\begin{align*}
S^{(N)}(t)
\approx & \unit_{2n} + \sum_{k=1}^{N |G|\cdot|\Gamma|} g_k \Big(\int_0^\tau \tilde{\gamma}_k(s)^{-1} A J \tilde{\gamma}_k(s) \rmd s \Big) g_k^{-1}\\
=& \unit_{2n} + \frac{t}{ \tau |G|\cdot|\Gamma|}\sum_{g\in G} g \Big( \sum_{\gamma\in\Gamma} \int_0^\tau \tilde{\gamma}(s)^{-1} A J \tilde{\gamma}(s) \rmd s \Big) g^{-1} \\
=& \unit_{2n} + \lambda J,
\end{align*}
with some $\lambda\in\R$ possibly depending on everything in the second line, as follows from the homogenization condition \eqref{eq:homogenizationcond} for $G$.

To summarize: if $G$ is a homogenization set generated by a minimal set $\Gamma$ then the modified procedure using a continuous implementation of Eulerian cycles of length $|G|\cdot|\Gamma|$ yields homogenization again but with finite-energy correction operations.

\subsection*{Decoherence suppression for system-environment interactions}

Now we want to come back to our initial motivation, the suppression of decoherence induced by generic quadratic system-environment interactions. The previous homogenization condition \eqref{eq:homogenizationcond} describes a mapping of interacting harmonic oscillators onto non-interacting oscillators that rotate with an averaged frequency. Hence in that condition the suppression of the interactions between the oscillators is included and it can be achieved by applying the operations \eqref{eq:homogenisationset} infinitely fast to the whole system. In general we have no access to the environment and therefore we now want to formulate a condition that allows us to suppress the system-environment interactions if we act on the system alone. We first partition the total system into the system of interest (S) with $n_S$ oscillators and the environment (E) with $n_E$ oscillators, so $n=n_S+n_E$, noticing that for symplectic dynamics we have a direct sum structure of the underlying space. 
Switching to the basis in \eqref{eq:basischange}, we thus write
\begin{align}\label{eq:envA}
A=
  \left(\begin{array}{c|ccc}
    A_{S} & ~ & I & ~ \\ \hline
    ~  & \multicolumn{3}{c}{\multirow{3}{*}{\raisebox{-0mm}{\scalebox{1}{$A_{E}$}}}} \\
    \raisebox{2pt}{$I^{T}$} & & &\\
    ~ & & & 
  \end{array}\right),
\end{align} 
where $A_{S}\in\Mat_{2n_s}(\R)$, $A_{E}\in \Mat_{2n_E}(\R)$ are symmetric matrices describing the uncoupled dynamics of S and E, and $I\in \Mat_{2n_S,2n_E}(\R)$ describes the interactions between system and environment.
Now, if we apply the decoupling operations only to the system, namely of the form $\tilde{g}=g\oplus\mathds{1}_{2n_{E}}$, we obtain in the limit of infinitely fast decoupling a dynamics governed by 
 \begin{align}\label{eq:tildePi}
\tilde{\Pi}(A)=\frac{1}{|\tilde{G}|}\sum_{\tilde{g}\in \tilde{G}}\tilde{g}A\tilde{g}^{T}
  =\left(\begin{array}{c|ccc}
    \frac{1}{|G|}\sum_{g\in G} gA_{S}g^T & ~ & \frac{1}{|G|}\sum_{g\in G}gI & ~ \\ \hline
    ~  & \multicolumn{3}{c}{\multirow{3}{*}{\raisebox{-0mm}{\scalebox{1}{$A_{E}$}}}} \\
    \raisebox{2pt}{$\frac{1}{|G|}\left(\sum_{g\in G}gI\right)^{T}$} & & &\\
    ~ & & & 
  \end{array}\right),
\end{align}
Obviously, in order to suppress the interaction with the environment, we need a group $G$ satisfying  
 \begin{align}
 \label{eq:decouplingdecsym}
 	\sum_{g\in G}g=0. 
 \end{align}
The simplest such group we can imagine is given by $G=\{\mathds{1}_{2n_{S}},~-\mathds{1}_{2n_{S}}\}$, and we notice that it leaves the system dynamics invariant. Thus, 
\[
\tilde{G}=\{\mathds{1}_{2n_S}\oplus\unit_{2n_E},
~-\mathds{1}_{2n_S}\oplus\unit_{2n_E}\},
\]
i.e., $\tilde{\Pi}(A)=A_{S}\oplus A_{E}$. 

The two operations in $G$ correspond to ``no-rotation'' and a global $\pi$-rotation of the system oscillators. It shows that the operation from \cite{Vitali}, introduced in the beginning, allows us to decouple arbitrary quadratic system-environment interactions too. This is not really surprising, since in the unitary picture we can always reverse the sign in front of interaction parts of the form $\hat{x}\otimes \hat{x}$ and $\hat{p}\otimes \hat{p}$ by applying $\exp(i\pi a^{\dagger}a)$.  Here however we want to emphasize two things. First of all, in contrast to finite-dimensional systems, the system can always be decoupled from the environment using two operations, independent of how big the system or the environment is.  Second, for finite-dimensional systems, on the one hand the irreducible action of the decoupling set suppresses all interactions with the environment, while on the other it modifies the system dynamics in such a way that it is (up to a global phase) given by the identity. For continuous variable systems described by quadratic Hamiltonians instead we can always suppress the interaction with the environment without disturbing the system dynamics at the same time. Let us summarize this as follows:

\begin{theorem}\label{th4}
Suppose a system of $n_S$ oscillators couples to an environment of $n_E$ oscillators and is such that the total time-evolution is described by a quadratic Hamiltonian. Then the interaction can be decoupled completely without influencing the system's internal dynamics, by choosing the decoupling set 
\[
\tilde{G}=\{\mathds{1}_{2n_S}\oplus\unit_{2n_E},
~-\mathds{1}_{2n_S}\oplus\unit_{2n_E}\} .
\]
\end{theorem}
Remark: since the homogenization set \eqref{eq:homogenisationset} satisfies Eq. \eqref{eq:decouplingdecsym} we can combine both procedures. The system can be homogenized while suppressing the interaction with the environment. Moreover, this can be adjusted again to Eulerian cycles and finite-strength correction operations.
%\vspace*{1cm}

\section{Random dynamical decoupling and bounds}

\subsection*{Probabilistic setup and derivation of bounds}
Up to now we have discussed how we can achieve homogenization and the suppression of decoherence for quadratic Hamiltonians in the limit of infinitely fast operations. Clearly, in practice this limit is not attainable meaning that non-zero orders in $\tau\| A  \|_\infty $ enter the dynamics. Throughout this article $\|\cdot\|_\infty$ denotes the standard operator norm on $\Mat_{2n}(\R)$ (or some other operator algebra depending on the context). In the following we provide bounds, characterizing how well dynamical decoupling works for continuous variable systems if the decoupling operations are applied reasonably though not infinitely fast. 

We discuss here the case of homogenization, while decoherence suppression can be treated analogously. Typically error estimates for dynamical decoupling are obtained by estimating the higher orders of the Magnus expansion, the Dyson series or the Trotter formula \cite{ErrorC}. Here we consider a random dynamical decoupling scheme \cite{LVioalRandD, USMath} and use a central limit theorem developed in \cite{USMath} in order to obtain the description of the modified time evolution $t\mapsto S^{(1)}(t)$ as a stochastic process in $\Sp(2n,\R)$. The idealized -- though impossible -- time evolution would result from $\tau=0$; it would be given by the (non-random) function
\begin{align}
\label{eq:targetevolution}
S_0(t)= e^{- t \Pi(A) J}, \quad t\in\R_+ ,
\end{align}
while $\Pi(A)$ has to be replaced by $\tilde{\Pi}(A)$ in \eqref{eq:tildePi} if we want to study decoherence suppression instead.

We would then like to find an upper bound for the expectation value of the gate error.
More precisely, we consider the case in which symplectic decoupling operations $g_{j}$ at time $j\tau$, with $j\in\N$, are taken independently and uniformly random from $G$ as in \eqref{eq:homogenisationset} such that the dynamics is modified according to 
\begin{align}
\label{eq:timeevolrandomdecd}
S^{(1)}(\ell\tau)=\prod_{j=1}^{\ell} g_{j}e^{-\tau AJ} g_{j}^{-1},
\end{align}
at time $t=\ell\tau$. Again, $g_j$ and $G$ would have to be replaced by $\tilde{g}_j$ and $\tilde{G}$ if we want to study decoherence suppression rather than homogenization.
The dynamics is now described by a random walk $\ell\in\N\mapsto S^{(1)}(\ell\tau)\in \Sp(2n,\R)$ on the symplectic group. In order to apply analytical tools, we would like to approximate this by a continuous-time stochastic process. This idea has been realized in \cite{USMath} for the case of dynamical decoupling. A similar treatment for homogenization or decoherence suppression, as shown in Appendix \ref{app}, leads to a limit stochastic process
\[
t\in\R_+ \mapsto S^{(\infty)}(t).
\]
If $\tau\|A\|_\infty\ll 1$ then this is a very good approximation of the actual time evolution $S^{(1)}(t)$. The quantity we would like to investigate then is the expectation of the \emph{gate error}:
\begin{align}
\label{eq:decouplingerror}
\varepsilon (t) =\|S_0(t)-S^{(\infty)}(t) \|_2^2,\quad t\in\R_+, 	
\end{align}
where $\|\cdot\|_2$ denotes the Hilbert Schmidt norm and $S_{0}(t)$ the idealized time evolution given by \eqref{eq:targetevolution}. We find:

\begin{theorem}\label{th5}
For $\tau\|A\|_\infty \ll 1$, i.e., if homogenization or decoupling operations are applied sufficiently fast, the expected gate error for homogenization behaves as
\begin{align}\label{eq:boundhom}
\E [\varepsilon (t)] \approx  2\tau t \| A-\Pi(A)\|_2^2 \le 16\tau t n \|A\|_\infty^2.	
\end{align}
In the case of decoherence suppression we get
\begin{align}\label{eq:boundsuppression}
\E [\varepsilon (t)] \approx 2\tau t \| A-\tilde{\Pi}(A)\|_2^2 \le 16\tau t n_{S}n_{B} k^{2}, 
\end{align}
with $k$ being the greatest absolute value of system-environment coupling entries in the matrix $I$ in \eqref{eq:envA}.
\end{theorem}
For a proof and more general treatment of bounds we refer to Appendix \ref{app}. Moreover, it would be interesting to study an adaption of the random procedure to Eulerian cycles as introduced in the context of homogenization, and to obtain explicit error bounds. Although we expect this to be feasible and meaningful, it goes far beyond the scope of the present article and provides interesting future work.

\subsection*{Simulations}
\begin{figure}
 (a)\phantom{\hspace{300pt}}
 \vspace{2pt}\\
 \includegraphics[width=0.7\columnwidth]{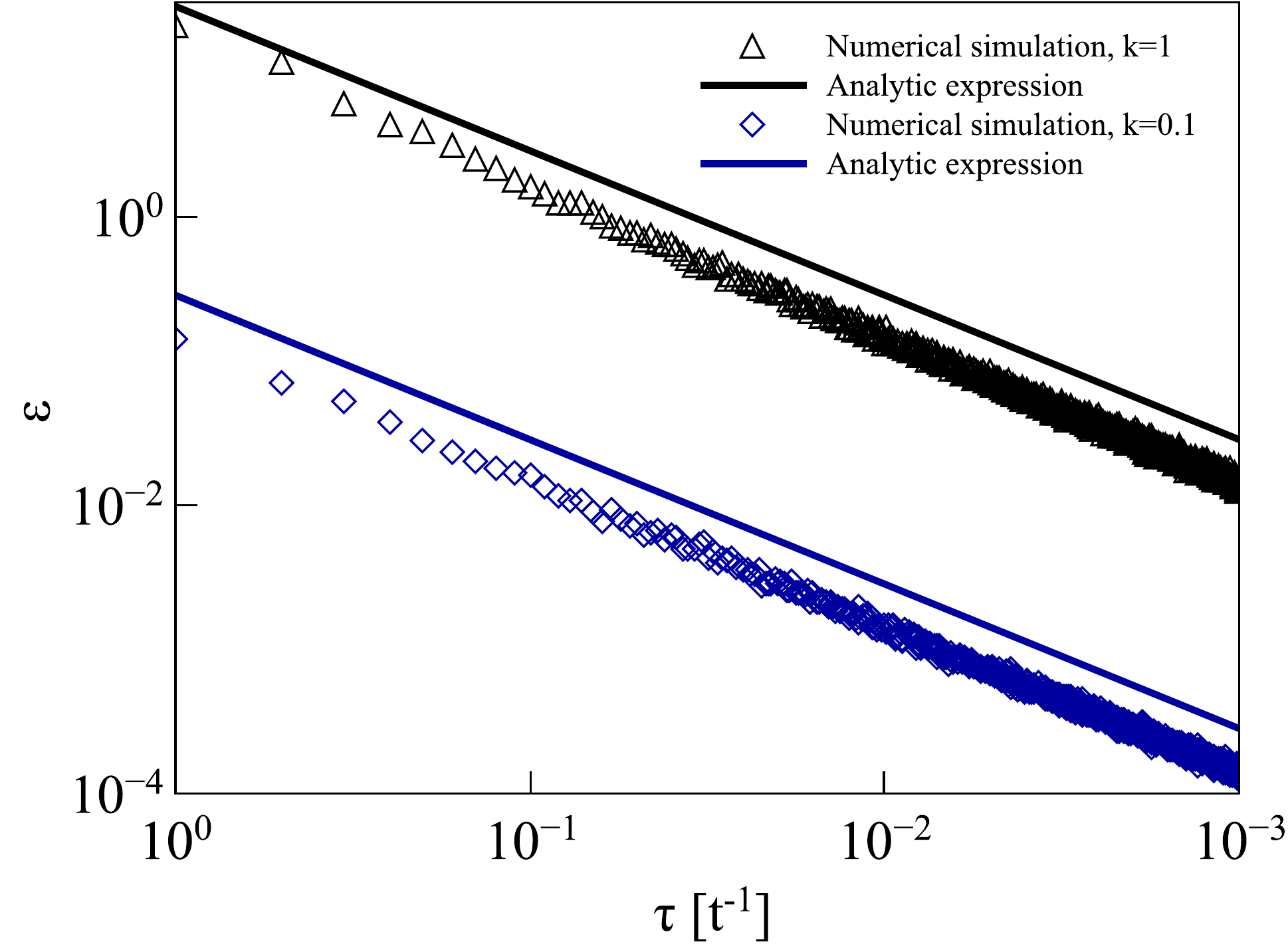}\\
 (b)\phantom{\hspace{300pt}}
 \vspace{2pt}\\
 \includegraphics[width=0.7\columnwidth]{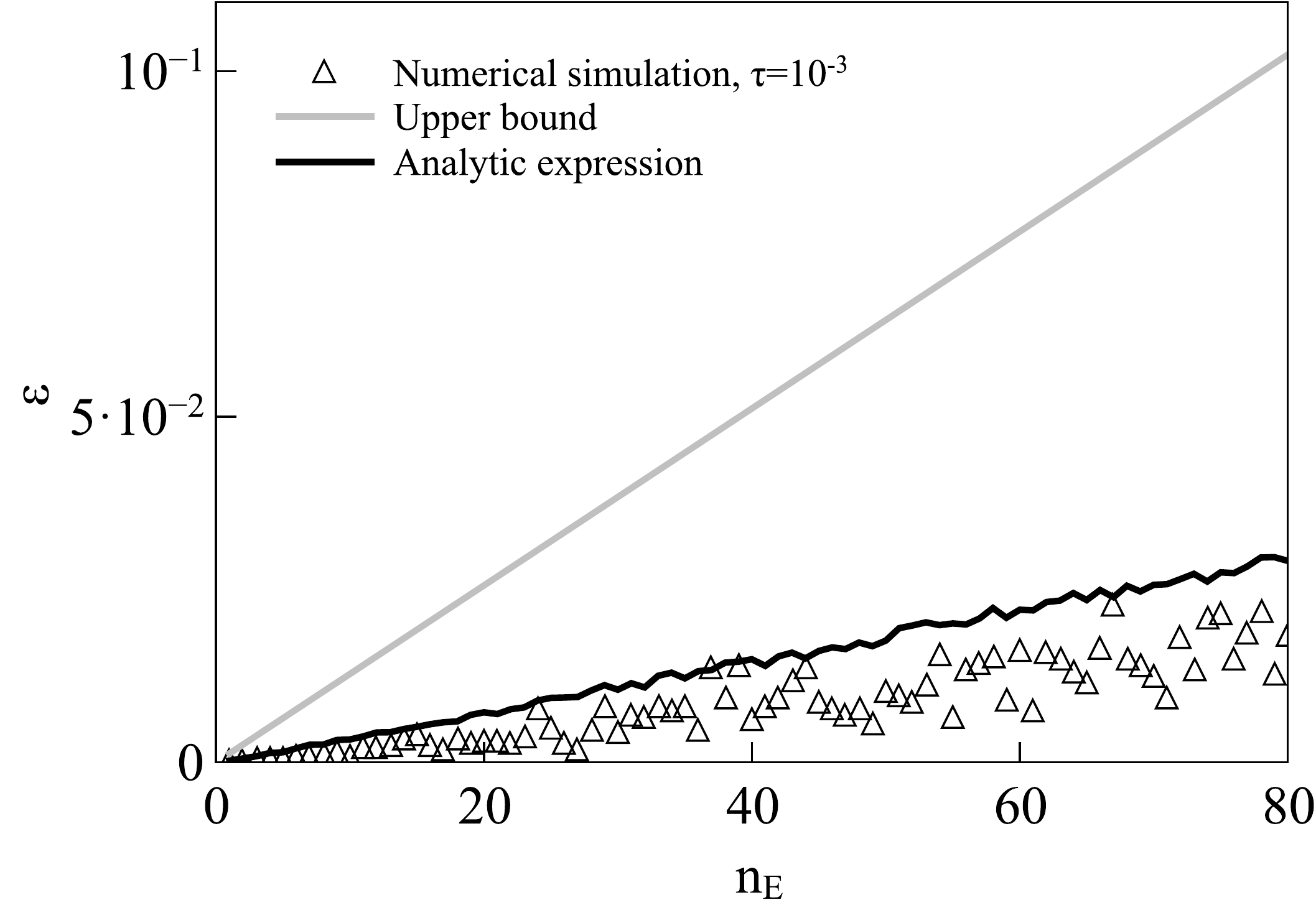}\\
 \caption{\label{fig:decouplingbounds} Numerical simulation of the averaged gate error for random dynamical decoupling. For homogenization (a) as function of the temporal spacing $\tau$ (on an inverse double logarithmic scale) between the decoupling operations and for the suppression of the system-environment interaction (b) as a function of the number of environmental oscillators $n_E$ for fixed $\tau=10^{-3}$. The average was taken over 20 trajectories for a total time $t=1$. (a): for a system of 4 interacting harmonic oscillators described by a randomly chosen matrix $A$ with entries bounded from above by $k=1$ (black triangles) and by $k=0.1$ (blue squares). The solid lines represents the corresponding analytic approximation depending on $A$ given by \eqref{eq:boundhom}. (b): for a system of two interacting harmonic oscillators coupled  to an environment consisting of $n_E=1,\ldots, 80$ oscillators (black triangles). The matrix $A$ describing the total system was chosen randomly with entries between 0 and 0.1. The solid lines show the more precise approximation depending on $A$ (black line) and the more general upper bound (grey line) in \eqref{eq:boundsuppression}, with $k=0.1$.}
\end{figure}
In order to illustrate and confirm the meaningfulness of our results, especially of Theorem \ref{th4}, let us now look at simulations for some standard models.

In Figure \ref{fig:decouplingbounds} we studied the validity of our analytic approximation in Theorem \ref{th4} by numerically evaluating the gate error \eqref{eq:decouplingerror} for homogenization (a) and for decoherence suppression of the system-environment interactions (b). In both plots we evaluated the gate error by taking the average over 20 trajectories that were obtained according to \eqref{eq:timeevolrandomdecd} for a total time $t=1$. For (a) the decoupling operations were taken independently and uniformly random from $G$ given by \eqref{eq:homogenisationset} and the expectation of the gate error was studied as a function of $\tau$. We investigated homogenization for $n=4$ interacting harmonic oscillators described by a randomly chosen matrix $A$ with entries between $0$ and $0.1$ (blue squares) and $0$ and $1$ (black triangles). In Figure \ref{fig:decouplingbounds} b) we studied the suppression of the system-environment interaction for a system of $n_S=2$ interacting harmonic oscillators as a function of the number $n_E$ of environmental oscillators for fixed $\tau=10^{-3}$ and taking the decoupling operations uniformly random from $G=\{\mathds{1}_{2n_{S}},~-\mathds{1}_{2n_{S}}\}$. The matrix $A$ describing the total system was chosen randomly with entries between $0$ and $0.1$. In both figures the solid lines represent the corresponding analytic expressions \eqref{eq:boundhom} and \eqref{eq:boundsuppression}. Remarkably, even for $\tau\Vert A \Vert\approx 1$ they describe very well the efficiency of the random decoupling scheme in terms of the temporal separation of the decoupling operations, the total evolution time and the entries of $A$. 

Exemplarily, a rough estimate can be given for optomechanical microresonators \cite{OptomechanicsExample}. Here the frequency of the mechanical oscillator and the cavity mode are in the Mhz regime and couplings of the order of kHz can be achieved between the mechanical oscillator and the cavity mode. The bound \eqref{eq:boundhom} suggests that for such systems $\tau \ll 1\,\mu\text{s}$ is required for homogenization of the motion of the mechanical oscillator and the cavity mode for a total time $t=1\,\mu\text{s}$. A time-scale analysis for decoherence suppression, for instance for a mechanical oscillator interacting with a bath of phonons, requires a microscopic model of the underlying decoherence mechanism. Particularly the order of the interaction between the mechanical oscillator and its environment or the spectral density of the environment need to be known.

\section{Conclusions}
We have studied dynamical decoupling for continuous variable systems that are described by quadratic Hamiltonians. We first proved that dynamical decoupling cannot work for every infinite-dimensional system. Thus, in contrast to finite-dimensional systems, not every infinite-dimensional system can be protected from decoherence using dynamical decoupling. Using the framework of symplectic transformations we afterwards investigated in more detail dynamical decoupling for quadratic Hamiltonians. We identified a condition and a set of operations that allows us to map a set of interacting harmonic oscillators onto non-interacting oscillators rotating with an averaged frequency. We called this process homogenization. Moreover we showed that every quadratic system environment-interaction can be suppressed with two simple operations without modifying the system dynamics at the same time. Using a random dynamical decoupling scheme we developed bounds characterizing how efficient both schemes are. We found that the efficiency depends on the temporal spacing of the decoupling operations, the total evolution time and the energy constants characterizing the considered system. Numerical simulations confirm the meaningfulness of the developed bounds.  

Our results pave the way for protecting an infinite-dimensional quantum system from decoherence which is induced by quadratic interactions. The reduction of noise in such systems is likely to find several applications. For instance, dynamical decoupling has the potential to decrease the environmentally induced errors in optical quantum computing \cite{OpticalQuantumC} and quantum metrology \cite{Metrology}. Moreover, dynamical decoupling for continuous variables might assist in verifying collapse models \cite{Bassi} in the macroscopic superposition regime \cite{MS1,MS2,MS3,MS4}. In particular, the reduction of noise caused by the interaction with the environment through dynamical decoupling might make the small derivations from the usual Schr\"odinger dynamics more visible \cite{USPRA}.

\appendix

\section{Construction of diffusion limit process and gate error bounds}\label{app}

The proof of Theorem \ref{th4} is similar in spirit to the diffusion limit theorem in \cite{USMath}, and we refer to that article for notation, motivation and procedure, in order to keep the present exposition succinct. We focus on the homogenization, and comment on decoherence suppression, which is treated analogously, towards the end.

To start with, we have to define the distribution of the increments in our random walk. A suitable choice is
\[
\mu_N := \frac{1}{|G|} \sum_{g\in G} 
\delta_{\exp\big( \frac{\tau}{\sqrt{N}}g (A-\Pi(A))J g^{-1} + \frac{\tau}{N}\Pi(A) J \big)}, \quad N\in\N ,
\]
This constitutes a family of measures on the Lie group $\Sp(2n,\R)$ with standard Borel $\sigma$-algebra, such that
\[
\mu_1 = \frac{1}{|G|} \sum_{g\in G} \delta_{\exp(\tau g A J g^{-1})}
\]
is the increment distribution of the actual time evolution operator $S^{(1)}(t)$ to $S^{(1)}(t+\tau)$, resulting from the instantaneous random homogenization operations as in \eqref{eq:timeevolrandomdecd}. 

We then consider the measures $(\mu_N)^{\ast N}$: an application of \cite[Th.3]{USMath}, cf. also \cite{Wehn}, shows that they converge to a normal distribution $\nu_1$ on $\Sp(2n,\R)$ as $N\ra\infty$, and $(S^{(N)}(t))_{t\in\R_+}$ converges to a Gaussian process $(S^{(\infty)}(t))_{t\in\R_+}$ on $\Sp(2n,\R)$, with distribution $(\nu_t)_{t\in\R_+}$ such that $\nu_0=\delta_{\unit_{2n}}$. The corresponding (dual) contraction semigroup has infinitesimal generator
\[
L = D_{\Pi(A)J} + \frac{\tau}{|G|} \sum_{g\in G} (D_{g (A-\Pi(A))Jg^{-1}})^2.
\]
The process $(S^{(\infty)}(t))_{t\in\R_+}$ cannot be expressed explicitly, but the expectation values of its matrix elements and higher moments can, thanks to the expression for $L$: we can apply it to the functions
\[
f_{kl}(g):= \langle e_k, g(e_l) \rangle, \quad g\in G,
\]
and
\[
f_{ij,kl}(g):= \langle e_i\otimes e_j, g(e_k)\otimes g(e_l) \rangle, \quad g\in G,
\]
where $(e_k)_{k=1,...,2n}$ forms an orthonormal basis of $\R^{2n}$.

More precisely, we get
\begin{align}\label{eq:genexp}
\E[\varepsilon(t)] =& \E[ \|S_0(t)-S^{(\infty)}(t) \|_2^2 \nonumber ]\\ \nonumber
=& \sum_{k,l=1}^{2n} \E\Big[ |\langle e_k, S^{(\infty)}(t)(e_l)\rangle|^2 
+|\langle e_k, S_0(t)(e_l)\rangle|^2 \\ \nonumber
& \quad -\langle e_k, S^{(\infty)}(t)(e_l)\rangle \overline{\langle e_k, S_0(t)(e_l)\rangle}
-\overline{\langle e_k, S^{(\infty)}(t)(e_l)\rangle} \langle e_k, S_0(t)(e_l)\rangle \Big]\\
=& \sum_{k,l=1}^{2n} \langle (e_k\otimes e_l), e^{t \hat{L}^{(2)}}(e_l\otimes e_k)\rangle + 2n \\ \nonumber
& \quad -\langle e_k, e^{t\hat{L}} (e_l)\rangle \langle e_l, S_0(-t)(e_k)\rangle
-\langle e_l, e^{t\hat{L}^\dagger}(e_k)\rangle \langle e_k, S_0(t)(e_l)\rangle, \nonumber
\end{align}
where
\[
\hat{L}= -\Pi(A)J + \frac{\tau}{|G|} \sum_{g\in G} \big(g (A-\Pi(A))Jg^{-1}\big)^2
\]
and
\begin{align*}
\hat{L}^{(2)}=& -\Pi(A)J\otimes \unit_{2n} - \unit_{2n}\otimes\Pi(A)J^\dagger\\ 
&+ \frac{\tau}{|G|} \sum_{g\in G} \Big( \big(g (A-\Pi(A))Jg^{-1}\big)^2\otimes \unit_{2n}
+ \unit_{2n}\otimes \big((g (A-\Pi(A))J g^{-1})^\dagger\big)^2 \\
& \qquad + 2 g (A-\Pi(A))Jg^{-1}\otimes (g (A-\Pi(A)) Jg^{-1})^\dagger \Big).
\end{align*}

Equation \eqref{eq:genexp} is the precise expression for $\E[\varepsilon(t)]$, which may actually be used for computer programs if $A$ is explicitly known and the dimension $n$ is reasonably small. In most other circumstances it makes sense to simplify \eqref{eq:genexp} under the physically realistic assumption that $\tau \|A\|_{\infty} \ll 1$. A first order expansion of all the exponentials in \eqref{eq:genexp} and the fact that $J^{-1}=J^\dagger$ and $g^{-1}=g^\dagger$, for all $g\in G$, lead to
\begin{align}\label{eq:expcalc}
\E[\varepsilon(t)] = & \nonumber
\frac{2\tau t}{|G|} \sum_{k,l=1}^{2n}\sum_{g\in G} 
\langle e_k, g(A-\Pi(A))J g^{-1} e_l \rangle \langle e_l, (g(A-\Pi(A))J g^{-1})^\dagger e_k \rangle + O(\tau^2 t^2 \|A\|_\infty^4)\\
\approx &  \frac{2\tau t}{|G|} \sum_{g\in G} \|g(A-\Pi(A))J g^{-1} \|_2^2\\
=& 2\tau t \|A-\Pi(A)\|_2^2.\nonumber
\end{align}
We notice that $\|\Pi(A)\|_\infty\le \|A\|_\infty$, so
\[
\|A-\Pi(A)\|_2^2 = \tr{(A-\Pi(A))^\dagger (A-\Pi(A))} \le \tr{ (2\|A\|_\infty)^2\unit_{2n}} = 8n\|A\|_\infty^2,
\]
which completes the proof of \eqref{eq:boundhom} for homogenization.

In the case of decoherence suppression instead, \eqref{eq:expcalc} becomes
\[
\E[\varepsilon(t)] \approx 2\tau t \|A-\tilde{\Pi}(A)\|_2^2.
\]
Now we note that 
\[
(A-\tilde{\Pi}(A))^\dagger (A-\tilde{\Pi}(A)) = \begin{pmatrix}
I I^T & 0 \\ 0 & I^T I
\end{pmatrix}.
\]
The maximum entry of $I I^T$ is bounded above by $2n_E k^2$, the maximum entry of $I^T I$ instead by $2n_S k^2$, with $k$ being the greatest absolute value of system-environment coupling entries in the matrix $I$. Thus
\[
\|A-\tilde{\Pi}(A)\|_2^2 = \tr{\begin{pmatrix}
I I^T & 0 \\ 0 & I^T I
\end{pmatrix}
}
\le 2n_S \cdot 2n_E k^2 + 2n_E \cdot 2 n_S k^2 = 8 n_S n_E k^2,
\] 
which proves the last part of Theorem \ref{th4}.

%\acknowledgments

%\bibliography{NoisyPlato}
%\end{document}

\end{document}